\documentclass[11pt]{article}

\usepackage{graphicx,natbib,amsmath,amssymb,enumerate,xcolor,multicol,amsthm,lineno,setspace}
\usepackage[shortlabels]{enumitem}
\usepackage[a4paper]{geometry}
\usepackage{bm}

\newtheorem{proposition}{Proposition}
\newtheorem{theorem}{Theorem}
\newtheorem{lemma}{Lemma}
\date{\today}

\newcommand{\bbeta}{\boldsymbol{\beta}}
\newcommand{\bphi}{\boldsymbol{\phi}}
\newcommand{\bz}{\mathbf{z}}

\newcommand{\bk}{\mathbf{k}}
\newcommand{\bX}{\mathbf{X}}
\newcommand{\bY}{\mathbf{Y}}
\renewcommand{\P}{\mathrm{P}}
\newcommand{\R}{\mathbb{R}}
\newcommand{\dd}{\mathrm{d}}
\newcommand{\E}{\mathrm{E}}
\newcommand{\Var}{\mathrm{Var}}

\newcommand{\rev}[1]{\textcolor{black}{#1}}

\title {Marked point processes intensity estimation using sparse group Lasso method applied to  locations of  lucrative and cooperative banks  in mainland France}
\author{Amélie Artis\footnote{Univ. Grenoble Alpes, CNRS, Sciences Po Grenoble, PACTE, France}, 
Achmad Choiruddin\footnote{Department of Statistics, Institut Teknologi Sepuluh Nopember, Indonesia}, \\
Jean-François Coeurjolly, Frédérique Letué \footnote{Univ. Grenoble Alpes, CNRS, LJK, France}}

\begin{document}

\maketitle

\paragraph{Short title.} Spatial analysis of banks locations

\paragraph{Acknowledgements.}  The authors are grateful to students from Master SSD at Univ. Grenoble Alpes who helped in collecting data, in particular the banks locations. 

\paragraph{Funding statement.} The research of JF Coeurjolly and F Letué is supported by Persyval-lab (ANR-11-61 LABX-0025-01). The research of A Choiruddin is funded by the Indonesian Endowment Fund for Education (LPDP) on behalf of the Indonesian Ministry of Higher Education, Science and Technology and managed under the EQUITY Program (Contract No 4299/B3/DT.03.08/2025 \& No 3029/PKS/ITS/2025).

\paragraph{Data availability statement.} Data used in this study are publicly available and may be freely reused, see~\citet{VVMZTN_2025}.

\paragraph{Conflict of interest disclosure and ethics approval statement.}The authors declare no conflict of interest and approve the ethics of the journal.



\newpage

\begin{abstract}
In this paper, we model the locations of five major banks in mainland France, two lucrative and three cooperative institutions based on socio-economic considerations. Locations of banks are collected using web scrapping and constitute a bivariate spatial point process for which we estimate non parametrically summary  functions (intensity, Ripley and cross-Ripley's $K$ functions). This shows that the pattern is highly inhomogenenous and exhibits a  clustering effect especially at small scales, and thus a significant departure to the bivariate (inhomogeneous) Poisson point process is pointed out. We also collect socio-economic datasets (at the living area level) from INSEE and propose a parametric modelling of the intensity function using these covariates. We propose a group-penalized bivariate composite likelihood method to estimate the model parameters, and we establish its asymptotic properties. The application of the methodology to the banking dataset provides new insights into the specificity of the cooperative model within the sector, particularly in relation to the theories of institutional isomorphism.
\end{abstract}

\noindent{\bf keywords:} cooperative banks; regional development; geo-referenced data; spatial statistics; point pattern analysis; regularization techniques; composite likelihood. \\

\section{Introduction}

The spatial distribution of banks is the outcome of a historical process. This process reveals that local banks established territorial oligarchies in the early 18th century \citep{plessis2000recherches}, which were later challenged by the formation of large national banks, such as Crédit Lyonnais in 1895. In addition to these family-owned institutions, cooperative banks emerged to offer financial services to individuals who were excluded from the traditional banking system. These local and retail banks were originally created to address the credit needs of workers, self-employed individuals, small farmers, and craftsmen who faced the burden of usury \citep{gueslin1998invention}. Founded in Europe during the 19th century, cooperative banks were structured around local communities, providing financial solutions managed by and for their members. 
Today, cooperative banks play a significant role in the banking and financial sector in several European countries. In 2020, they accounted for over 30\% of the domestic deposit and loan market in Austria, Finland, France, and the Netherlands \citep{eacb}. Their funding activities have been noted for producing idiosyncratic knowledge \citep{artis2016composition}. Furthermore, the proximity between lenders and borrowers enhances the availability of credit and reduces default rates \citep{fisman2017cultural}. However, as cooperative banking groups increasingly concentrate and consolidate \citep{ory2012efficiency}, the distinctive characteristics of cooperative banks seem to be diminishing. This trend distances them from their local communities, which have been essential to their success for many years \citep{roux2019introduction}. The loss of democratic engagement, the rise of universal banking, and the presence of institutional investors in international financial markets raise the question: Do significant differences still exist? 
Currently, competition and new methods of delivering banking services, particularly digitalization, challenge the importance of a local presence. The digitalization of banking services and the emergence of online banks are intensifying competition as consumer behavior evolves, with local physical presence becoming less of a priority in decision-making. Increased competition and the move toward digital services have led to a rationalization of physical locations. However, does this mean that the unique spatial characteristics of banks have been eliminated?

As part of a geographical agenda, this article outcomes and examines a mapping of the spatial distribution of cooperative banks and searches for the reasons underlying this point pattern. We develop a spatial analysis to identify the possible existence of regional differentiation between cooperative and non-cooperative banks.
Our work raises the following questions: How are lucrative and cooperative banks distributed across mainland France?
To what extent is location still a differentiating factor between cooperative and non-cooperative banks?
What does the new location landscape look like for cooperative banks?
To what extent can current bank locations be explained by socio-economic factors?
Are these economic factors still differentiating factors for cooperative banks?
We aim to investigate whether a spatial differentiation exists between profit-driven banks and cooperative banks. Analyzing spatial variables offers a relevant research opportunity. Firstly, there have been few studies capable of spatially representing the geographical distribution of banks. Secondly, spatialization is one of the essential conditions of the concept of local banking. Finally, geographical presence is necessary for the production of idiosyncratic knowledge, which relies on geographical proximity \citep{artis2016composition}.

\section{Literature review and research questions}

Since the establishment of the single European market, the convergence of interest rates and the introduction of the universal banking model have led to the dismantling of geographical restrictions in the banking sector. This shift has prompted banks to expand their networks and standardize their offerings, with centralization reducing the autonomy of local branches. Consequently, the central argument of this analysis is that these changes have led banks to shift their strategic focus toward larger markets, thereby reducing their emphasis on local concerns.

For cooperative banks, the number and location of branches have traditionally been central to their strategic differentiation. Cooperative banks have their origins in local communities, arising out of the lack of a conventional banking system in France \citep{artis2013finance}. This historical context suggests a distinct spatial configuration \citep{le2019geographie}, where the banks’ operations were tailored to the needs and dynamics of their local environment. 
\rev{The spatial distribution of bank branches at any given date results from the accumulation of location decisions made over time. These decisions are shaped by long-term strategy as well as short-term real estate opportunities. Cooperative and non-cooperative banks use broadly similar procedures and timeframes to establish new branches; however, their location criteria differ. Specifically, cooperative banks have historically favoured proximity to their member base and local economic activity, whereas non-cooperative banks often prioritise commercial catchment potential. Given these differing criteria, our analysis captures the current spatial configuration of branch networks and models it as a function of present territorial characteristics.}

\rev{Our {\bf first hypothesis} posits that spatial differentiation may still be present in the current distribution of bank branches. Our {\bf second hypothesis} is that underlying factors may explain this spatial distribution.}

Literature on the relationship between the banking system and regional development highlights several key issues: whether local banking systems contribute to local economic growth; differences in the spatial configurations of banking networks; the factors that influence the decision to maintain or withdraw a local banking presence. The characteristics of local banking markets are important for local economies, especially in regard to small and medium-sized enterprises (SMEs). Geographic distances between banks and businesses affect access to credit, with small businesses particularly impacted by such distances. \citet{presbitero2014home} show that most loans are contracted locally, reinforcing the importance of proximity. Similarly, \citet{liberti2008estimating} and \citet{agarwal2010distance}  argue that a shorter geographical distance between the information-gathering agent and the loan officer enhances the collection and utilization of soft information about potential borrowers. In contrast, \cite{mian2006distance} suggests that greater distances complicate contract renegotiations and reduce the likelihood of successful loan recovery.

\citet{alessandrini2009global} further emphasizes the significance of local credit markets for financing small borrowers and fostering local economic development. \citet{reydet2019changement} identifies emerging branch models that blend high digitalization with physical presence, aiming to sustain a geographical network of locations.

In the Italian context,  \citet{papi2017geographical} explores the geography of banking organizations, focusing on how geographical distances influence the banking sector. The study identifies three types of distance: (i) operational distance, defined as the distance between the borrower and the bank branch; (ii) functional distance, which refers to the distance between the branch and the bank’s headquarters; and (iii) interbank distance, which denotes the distance between borrowers and competing banks. By applying the theory of banking centers and peripheries, \citet{papi2017geographical} demonstrates that these geographical distances significantly impact lending decisions and interbank competition.

Much of the existing literature supports the argument that cooperative banks possess a comparative advantage in lending to SMEs \citep{artis2016composition}. \rev{Funding applications typically rely on quantitative ('hard') information, primarily financial data and standardised data, as well as qualitative ('soft') information, which is more personalised \citet{berger2002small}.}
The presence of a local bank that prioritizes relationship banking and utilizes soft information has been shown to positively influence local economic activity. According to \citet{hasan2019economic}, the development of local bank branch networks plays a vital role in stimulating economic activity, particularly in the local labor market.

\rev{Based on this literature, the two types of banks serve different customer bases, and their spatial distribution may also reflect these distinct market orientations. Building on this established evidence, our analysis examines whether the observed spatial patterns of branch locations are consistent with such differentiated strategies — particularly in terms of proximity to areas characterised by different household income levels and business demographics}. 

For a long time, the spatial distribution of banks was guided by local market theory \citep{radecki1998expanding}, which posits that customers—both businesses and households—primarily consume banking services based on geographical proximity. In addition to this proximity, the distribution of banks also follows the logic of customer segmentation. Mutual banks, such as Crédit Mutuel and Crédit Agricole, have historically focused on meeting the needs of farmers and small businesses \citep{artis2016composition}. They tend to establish branches close to their members and customers. In contrast, major banking networks like BNP Paribas and Société Générale develop urban branches to attract a diverse customer base while maintaining regional decision-making centers.
\citet{radecki1998expanding} argues that the spatial distribution of banks is now driven more by organizational logic from the supply side rather than by customer behavior on the demand side. This shift implies the disappearance of local markets and the obsolescence of the local bank model, including cooperative banks as the archetypal representation of local banks. These changes in banks’ offerings suggest that the previously distinct spatial distribution among banks has diminished.
In these analyses, the organizational specifics of banks, such as their cooperative principles and rules, are not treated as significant explanatory variables. \citet{algeri2022spatial} investigate the impact of spatial dependence on the technical efficiency of small cooperative banks in Italy. They explore the idea of spatial dependence, suggesting that banks within a specific geographical region can influence one another. For cooperative banks, their findings reaffirm that geographical proximity is critical to bank performance.

\section{Methodologies}

\rev{Our analysis focuses on}  the case of France for theoretical and empirical reasons. First there has been a large and strong cooperative banking sector for two centuries. Then, the French banking and financial sector is considered to be well-established, offering comprehensive coverage and stable banking and financial companies, as well as an effective regulatory system.  Finally, there are homogeneous and longitudinal statistical data produced by the National Statistical Institute. 

The French banking system consists of two distinct groups of banks, which differ significantly in their historical development, business models, organisational structures, and scale. The first group, referred to as lucrative banks, includes institutions organised as joint-stock companies or as branches of foreign credit institutions operating nationwide branch networks. The second group comprises only cooperative banks, which are structured by regional clusters and are owned by their members.
The six largest French banking groups (BNP Paribas, BPCE, Crédit Agricole, Crédit Mutuel, La Banque Postale and Société Générale) account for 82\% of the total balance sheet of the entire French banking sector and are considered systemic by the regulator \citep{acpr}. In the last quarter of 2018, cooperative banks managed more than €159 billion in demand deposits from individuals, representing more than 40\% of total deposits. Cooperative banks finance approximately 41\% of consumer loans to residents (Web Stat, Banque de France). 

\rev{We selected five banks, within the six banking groups designated as systemically important institutions (établissements d'importance systémique) by the French prudential supervisory authority (ACPR, 2024). There are two non-cooperative banks and three cooperative banks thereby ensuring balanced representation of both governance models among the most structurally significant actors in the French banking system. {One group was excluded: La Banque Postale, although classified as a non-cooperative bank, was excluded because its branch network originates from the postal service infrastructure rather than from a market-driven location strategy, which would introduce a confounding factor in our spatial analysis} Together, the five selected groups offer nationwide geographical coverage. This coverage is particularly relevant given that a central objective of our study is to assess whether the historically documented stronger presence of cooperative banks in rural areas remains a statistical reality in the contemporary distribution of branches, or whether the spatial strategies of cooperative and non-cooperative banks have since converged. Examining groups with a genuinely nationwide footprint is therefore a necessary condition for this analysis}.

At the level of the social economy, including cooperatives, works use the combination of cluster analysis with multiple factor analysis (MFA), e.g. in Italy \citep{picciotti2014social}  or in France \citep{artis2020facteurs}. For cooperative banks, spatial distribution structures the activity and value chain of the cooperative, as well as its governance \citep{triboulet2013empreinte}. For example, \citet{triboulet2013empreinte} characterise the location patterns of agricultural cooperatives in France for 1995 and 2005 in order to better understand the determinants of their spatial organisation. They use the ESDA tools to construct a very detailed local typology of the specific location patterns of agricultural cooperatives, and then analyse the explanatory factors using multinomial logit modelling. Our analysis consists of two parts:  
\begin{itemize}
\item[(a)] \textit{Non parametric spatial analysis of locations of bank branches (lucrative and cooperative banks)}. \rev{We have two objectives in this analysis which focuses on answering our first hypothesis. First, we investigate the spatial variation of bank locations over the French territory, individually for the locations of lucrative and cooperative banks. Second, we explore eventual dependence between locations of two lucrative banks, between two cooperative ones and between one of each. To achieve these tasks, we collected exact GPS locations of 2 lucrative banks (BNP, Société Générale) and 3 cooperative ones (Banque Populaire, Crédit Agricole, Crédit Mutuel) in mainland France, as of 01/01/2022. These five banks cover 92\% of bank clients in France (as of 2023). Locations have been obtained using GPS transformations of banks addresses collected using web scraping. This results in 3,448 lucrative banks and 11,244 cooperative ones.}

\item[(b)] \textit{\rev{Model} of the intensity of each bank branch.} The objective is to propose a bivariate parametric modelling of the intensity function of lucrative and cooperative banks in terms of demographic, social and economic factors. \rev{Such a modeling framework proposes a better understanding on how demographic, social, and economic factors influence the distribution of locations of lucrative and cooperative banks, answering our second hypothesis. We rely on a dataset collected  by INSEE} as of 01/01/2022 at the living area spatial resolution (in total 1681 living areas). This dataset, after transformations taking into account the compositional nature of some covariates \citep{aitchison1982statistical}, results in eleven spatial covariates: three of them, tagged as demographic covariates, describe the population by ages and their recent evolution, four covariates tagged as social covariates are related to the standard of living, rate of poverty, rate of activity. Finally, four covariates describe the economical tissue. We also collect the \rev{population density} by living area. The objective of this study is to detect factors that  positively or negatively influence the presence/absence of cooperative or lucrative banks.  
\end{itemize}

\rev{As a summary, to achieve analyses (a)-(b), we created two datasets. The first one consists in  geographical coordinates} for five banking networks in France. While this is straightforward for lucrative banks, which have centralised management of their development across the country, this data does not exist in a consolidated form for cooperative banks, which manage their development on a regional basis, as autonomous banks. \rev{The second one is} produced by the National Statistics Agency \rev{(INSEE) and is collected at the living area level.} 
\rev{These area-level variables — capturing demographic composition, household income characteristics, and the local business fabric (including firm size distribution) — serve as covariates in our bivariate parametric model of the intensity functions of non-cooperative and cooperative bank branches. As seen in Section~\ref{sec:parametric}, in this framework, we model the spatial intensity of each bank type as a function of the demographic, social and economic characteristics of the surrounding bassin de vie, allowing us to identify which territorial features are associated with a higher or lower density of each type of bank branch.} All data used in this study are publicly available and may be freely reused, see~\citet{VVMZTN_2025}.

The rest of the paper is structured as follows. Section~\ref{sec:model} provides first a background on spatial point processes and investigate non parametric first and second-order intensity functions for the dataset of locations of banks. This section concludes that both cooperative and lucrative banks are highly inhomogeneously distributed and exhibit a dependence structure—both within each type and between the two types of banks, that cannot be explained solely by inhomogeneity.
Accordingly, Section~\ref{sec:parametric} investigates a parametric specification of the intensity functions for the two bank types (using an exponential family model), incorporating three groups of spatial covariates: (i) demographic characteristics, (ii) indicators of living standards and poverty/activity rates, and (iii) measures of the local economy. Given, the nature of these covariates and the bivariate multitype characteristic of the dataset, we propose to estimate the parameters using a new methodology based on a sparse group lasso penalization of the Poisson likelihood (a standard composite likelihood method). Section~\ref{sec:results} then presents the numerical implementation of the method on the banks dataset, reports the empirical results, provides a discussion, and concludes the analysis. Finally, Appendices~\ref{app:prop}–\ref{app:convergence} provide formal justification for the simplicity of the proposed implementation and establish convergence results, demonstrating that the method is consistent and reliably selects informative covariates when the amount of data increases.

\section{Spatial point processes and non parametric considerations} \label{sec:model}

\subsection{Background and notation} \label{sec:bck}

In this manuscript, we use the convention $[\{(\dots)\}]$ for the order of brackets. By nature, the number of banks as well as their locations are unknown before collecting the data. Therefore, the dataset of locations of banks is viewed as the realization of a bivariate spatial point process observed on, say $W \subset \mathbb R^2$, representing the mainland French territory. Spatial point processes are now well-studied stochastic models and we refer the interested reader by a theoretical presentation to~\citet{moller2003statistical} or by a more practical one to~\citet{baddeley2015spatial}. In particular, the latter discusses the implementation of point pattern analysis in \texttt{R}.

Let $\mathbf X_\ell$ and $\mathbf X_c$ be the two  point processes for lucrative and cooperative banks modelling these data. Their realization are of the form
\[
\mathbf x_\ell = \left\{ x_{i,\ell}, i=1,\dots,n_\ell\right\}
\quad \text{ and }\quad
\mathbf x_c = \left\{ x_{i,c}, i=1,\dots,n_c\right\}
\]
where $x_{i,\ell},x_{i,c} \in W$ denote the locations of a lucrative or cooperative bank {with total respectively} $n_\ell=3,448$ and $n_c=11,244$. We let $\rho_\ell$ and $\rho_c$ denote the respective intensity function of each type of bank. This is defined for $\bullet=\ell,c$ by $\rho_\bullet(x) = \lim_{\mathrm{d} x\to 0} \mathbb E N_\bullet(B(x,\mathrm dx)) /|\mathrm dx|$ where $N_\bullet(A)$ is the number of points of type $\bullet$ falling into $A\subset \mathbb R^2$. Roughly speaking, the intensity function measures the mean local \rev{number} of points per unit area. The higher the intensity at a location $x$, the most likely a bank should occur  at $x$. A point process is said to be inhomogeneous if the intensity function varies across  space.

In addition to the inhomogeneity nature of data, we could be interested in the dependence between two points (banks). To measure departure \rev{from the independence case, often represented by the bivariate Poisson point process,} which is the standard process modelling points (possibly inhomogeneous) without any interaction, spatial statistical literature provides many standard summary functions. Among them, we find the pair correlation \rev{and its integrated version named as the Ripley's $K$ function. For a stationary point process, $K(r)$ for some $r\ge 0$ measures the normalized (by the intensity) mean number of points with a ball centered at 0 with radius $r$, given that 0 is a point from the point process. Departures with the Poisson point process can be measured since, in this case, $K(r)=\pi r^2$ (for two-dimensional point patterns). To measure correlation between two point processes, we can inspect the cross-pair correlation function and its integrated version the cross-Ripley's $K$ function. We refer the reader to~\citet{moller2003statistical,baddeley2015spatial} for a more formal definition, and for their extensions to the inhomogeneous case. It is worth mentioning that the Ripley's and cross-Ripley's $K$ still equal $\pi r^2$ for bivariate Poisson point processes. 
The $L$ and $\tilde L$ functions are simple transformations of $K$ functions, given by $L(r) = \sqrt{K(r)/\pi}$ and $\tilde L(r) = L(r)-r$.
Hence,these summary functions measure within a distance $r$ whether there is a significant departure to the Poisson model. 
In the following, we will denote by $K_\bullet, L_\bullet, \tilde{L}_\bullet, \bullet=c,\ell$ Ripley's $K$, $L$ and $\tilde L$ functions for the cooperative ($\bullet=c$) banks point process and the lucrative ($\bullet=\ell$) banks point process, and by $K_{c,\ell},  L_{c,\ell}, \tilde{L}_{c,\ell}$ cross-Ripley's $K$, $L$ and $\tilde L$ functions.
$\tilde L_\bullet(r)>0$ (resp. $\tilde L_\bullet(r)<0$), suggests a significant clustering (resp. repulsive) effect, not explained by the inhomogeneity, between two banks of type $\bullet$ within distance $r$.} In the same way, $\tilde L_{c,\ell}(r)>0$ reflects the fact that, within distance $r$, one type of bank  influences positively the occurrence of the other one.

\subsection{Non parametric estimation of first and second-order characteristics} \label{sec:nonparametric}

Figures~\ref{fig:int}-\ref{fig:L} summarize this part and explore first and second-order characteristics of the bivariate point process. In Figure~\ref{fig:int}, we compute a kernel density \rev{estimate} of the intensity function. We consider adaptive smoothing bandwidths  (using a 16x16 spatial pixel grid for the bandwidth spatial function and a Gaussian kernel) according to the inverse-square-root rule of \citet{abramson1982bandwidth}. Even if three times more cooperative banks are observed, we clearly notice large differences in terms of spatial distribution between the two \rev{banks of different type. Lucrative banks are clearly concentrated in Ile-de-France and the largest cities (Lille, Lyon, Marseille, Toulouse, Bordeaux) while cooperative banks aim at covering more territory especially in the Brittany (Western part), Alsace-Lorraine (Eastern part) and the Région Centre.}

Figure~\ref{fig:L} explores second-order properties and represents empirical inhomogeneous recentered $\tilde L$ functions for each type of bank and cross $\tilde L$-function between the two types of banks. The solid line represents the estimate for which the preliminary estimation of the intensity function has been integrated. Envelopes have been constructed under the inhomogeneous Poisson assumption using the global envelopes testing machinery proposed by~\citet{myllymaki2017global}. We use the extreme rank length criterion and 199 simulations under the null assumption. Red dots correspond to distances where a significant departure to the Poisson assumption is observed. It is unambiguous to conclude that first cooperative and lucrative banks exhibit a quite strong clustering effect, especially at small distances, that is not explained by their inhomogeneity. In the same way, the cross $L$-function at small distances shows that two banks different type tend to occur more likely than under the Poisson case. This {reveals a strong} co-clustering at small scales, that is not explained by the inhomogeneity. 

\begin{figure}[ht]
    \centering
\includegraphics[width=.99\textwidth]{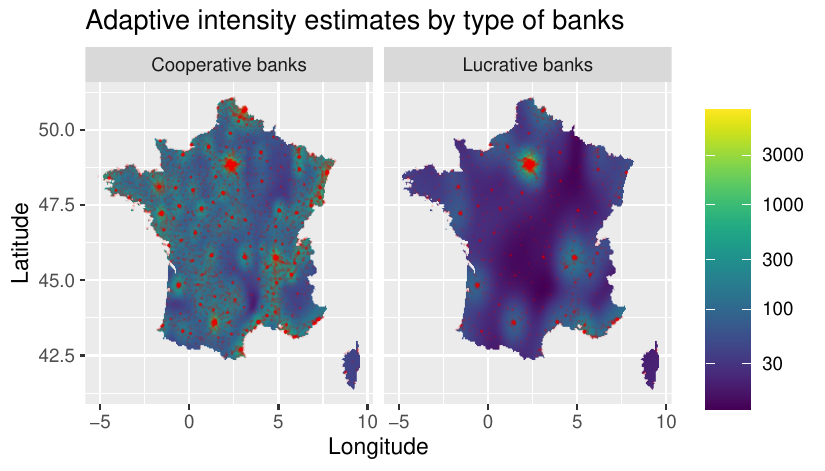}
    \caption{Adaptive intensity estimates obtained from kernel method with standard edge correction, for cooperative and lucrative banks. Locations of banks (in red) are superimposed on the intensity estimate. \rev{The pixel resolution of both images is 256x256.}}
    \label{fig:int}
\end{figure}

\begin{figure}[ht]
    \centering
\includegraphics[width=\textwidth]{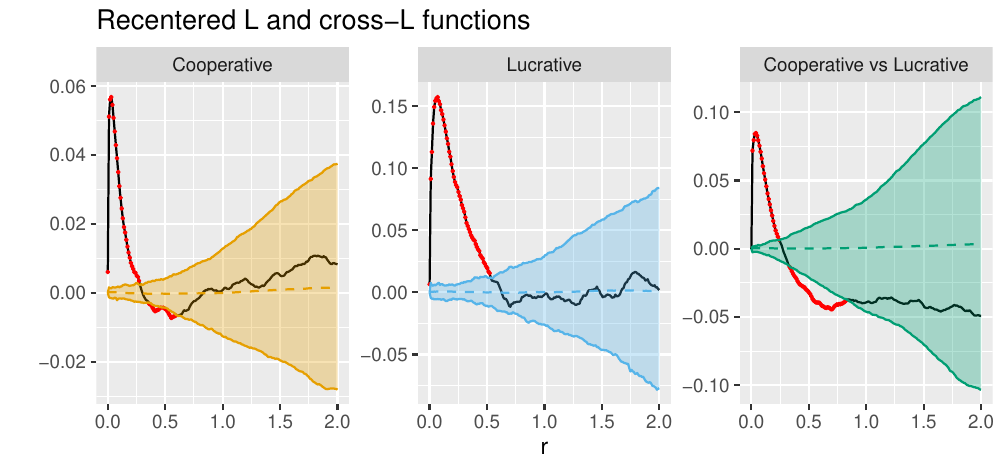}
    \caption{Solid lines represent estimates of $\tilde L_\bullet(r)$ (for $\bullet=c,\ell$) and $\tilde L_{c,\ell}(r)$ \citep{baddeley2015spatial}. Envelopes represent 95\% global envelopes constructed using the extreme rank length measure \citep{myllymaki2017global} under the inhomogeneous Poisson assumption. Red dots indicate distances for which a significant departure to the Poisson assumption is observed \rev{(the x-axis $r$ corresponds to a ball with radius $r$ with coordinates expressed in longitudes/latitudes)}.}
    \label{fig:L}
\end{figure}

\section{Parametric intensity estimation} \label{sec:parametric}

We now turn to the main core of this study which is to model $\rho_\bullet$ for $\bullet=c,\ell$ using covariates observed at the living area level. To make this paper more general, we recode $1=c$ and $2=\ell$ respectively for cooperative and lucrative banks (so that we can consider higher dimension than bivariate point process).

For computational arguments we consider the following exponential family model for each function
\begin{equation}
\rho_i (x; \bbeta_i)  = d(x)\exp 
\Big\{
\bbeta_{i}^\top \bz_i(x)
\Big\}
\label{eq:int}    
\end{equation}
where $\bz_i(x) = \{z_{j}(x)\}_{j=1,\dots,b_i}$  represents the vector of spatial covariates at location $x\in W$ related to cooperative (for $i=1$) and lucrative banks (for $i=2$). These covariates are discussed later. The vector $\bbeta_i=\{(\bbeta_{i})_1,\dots,(\bbeta_i)_{b_i}\}^\top \in \mathbb R^{b_i}$ represents the parameter vector of interest for cooperative banks when $i=1$ and for lucrative banks when $i=2$. We let $(\bz_i)_1(x)=1$ for any $x$, so that $(\bbeta_i)_1$ represents the intercept term for the bank of type $i$. \rev{Finally, the covariate $d(\cdot)$ represents population density at location $x$. We include it as a baseline function}, as our primary objective is to compare the two types of banks with respect to socio-economic factors rather than population concentration.

We can also rewrite~\eqref{eq:int} as a marked point process intensity model
\begin{equation}
\rho [\{(x,i); \bbeta\}]  = d(x) \prod_{j=1}^2
 \exp 
\Big\{ \delta_{ij} \; \times \bbeta_j^\top \mathbf z_j (x)
\Big\}
\label{eq:marked}    
\end{equation}
where $x\in W$ and $i\in\{1,2\}$ where $\delta_{ij}=1$ when $i=j$ and 0 otherwise.


The  covariates we consider for our dataset application are \rev{identical for the two types of banks. These are} categorized into demographic, social and economic.  \rev{Note that all these covariates, collaected as a single INSEE file as of 01/01/2022), as well as the baseline population density $d(\cdot)$,  are piecewise-constant over space, as they are measured at the living area spatial resolution.}

\begin{itemize}
\item[(i)] \textit{demographic covariates}. \texttt{evolution}: \rev{average annual rate of change of the total resident population of the living area between 2014 and 2020. The next covariates are related to age, which is indeed a marker of banking behaviour linked to changes in working life. For example, people over the age of 65 tend to have more assets or savings than younger people. Conversely, young people often follow their parents' lead when opening their first bank account. The age criterion is also examined to understand the shift from physical to online banks (Ngau et al., 2023; Tesfom et al., 2011). These age variables are \texttt{prop16.24}, \texttt{prop25.64} and \texttt{prop65.more}: \rev{average annual rate of change of the total resident population of the living area between 2014 and 2020, in \%  between 16 and 24 years old, etc}. }To deal with these compositional data we consider the additive log ratio transform - alr \citep{aitchison1982statistical} taking the last variable as the reference one. The transformed variables are  \texttt{log(prop16.24/prop65.more)} and \texttt{log(prop25.64/prop65.more)}.
\end{itemize}
\begin{itemize}
\item[(ii)] {\it social covariates}. \rev{According to previous work, we select social covariates to qualify the income distribution (poverty and activity) and another to qualify social inequality (median and decile ratio) by local zone.
These two sets of covariates are chosen to explain the segmentation of bank clients and document the influence of the demand side. The first two variables are} \texttt{median}, \texttt{decile.ratio}: median, inter-decile ratio of the standard of living in the living area; \texttt{activity}, \texttt{non.activity}: \rev{labour-force activity rate, i.e. share of the working-age population (15–64) that is economically active (employed or unemployed but seeking work), in \%}; \texttt{poverty}, \texttt{non.poverty}: {poverty rates in the living area and their complement to 1}. Again we use the alr transform and consider  the  variables \texttt{log(activity/non.activity)} and \texttt{log(poverty/non.poverty)}.
\end{itemize}
\begin{itemize}
\item[(iii)] {\it economic covariates}. \rev{The third set of covariates, the economic covariates, aims to describe the types of businesses and forms of employment present in a given area. This information helps banks understand their customers better. Therefore, we select relevant economic covariates to help us qualify the area, such as the importance of SMEs and the types of work in the local economy.} \texttt{prop.industry}: {proportion of positions in industry}; \texttt{prop.public}: {proportion of positions in public administration, teaching, health or social action}; \texttt{prop.trade}: {proportion of positions in trade, transport or commercial sector}; \texttt{prop.agri}: {proportion of workers in the farming sector plus proportion of positions in the construction sector} \footnote{The  notation \texttt{prop.agri} is justified by the fact that the proportion of positions in the construction sector is negligible compared to the proportion of workers in the farming sector}. These four variables sum to 1 and are thus transformed using the alr transform, which yields the three variables \texttt{log(prop.industry/prop.agri)}, \texttt{log(prop.public/prop.agri)} and 
\texttt{log(prop.trade/prop.agri)}; the last two covariates are \texttt{prop.1to9} and \texttt{prop10.more}: {proportion of companies with 9 or less employees, resp. more than 10 employees}. These covariates are transformed using the log ratio transform into the single covariate \texttt{log(prop.1to9/prop10.more)}
\end{itemize}

Thus we come up with three demographic covariates, four social ones and four economic ones,  illustrated in Figure~\ref{fig:cov}. To estimate the intensity model~\eqref{eq:marked}, we have therefore for each mark 1 and 2, one intercept parameter and eleven parameters corresponding to spatial covariates. Therefore, $\bbeta_i \in \mathbb R^{b_i}$ with $b_i=12$ for $i=1,2$ and so $\bbeta\in \mathbb R^p$ with $p=\sum b_i=24$. It is worth mentioning that we consider the same set of covariates for both marks but the methodology presented hereafter could be reproduced for more general situations.

\begin{figure}[htbp]
    \centering
    \includegraphics[width=\textwidth]{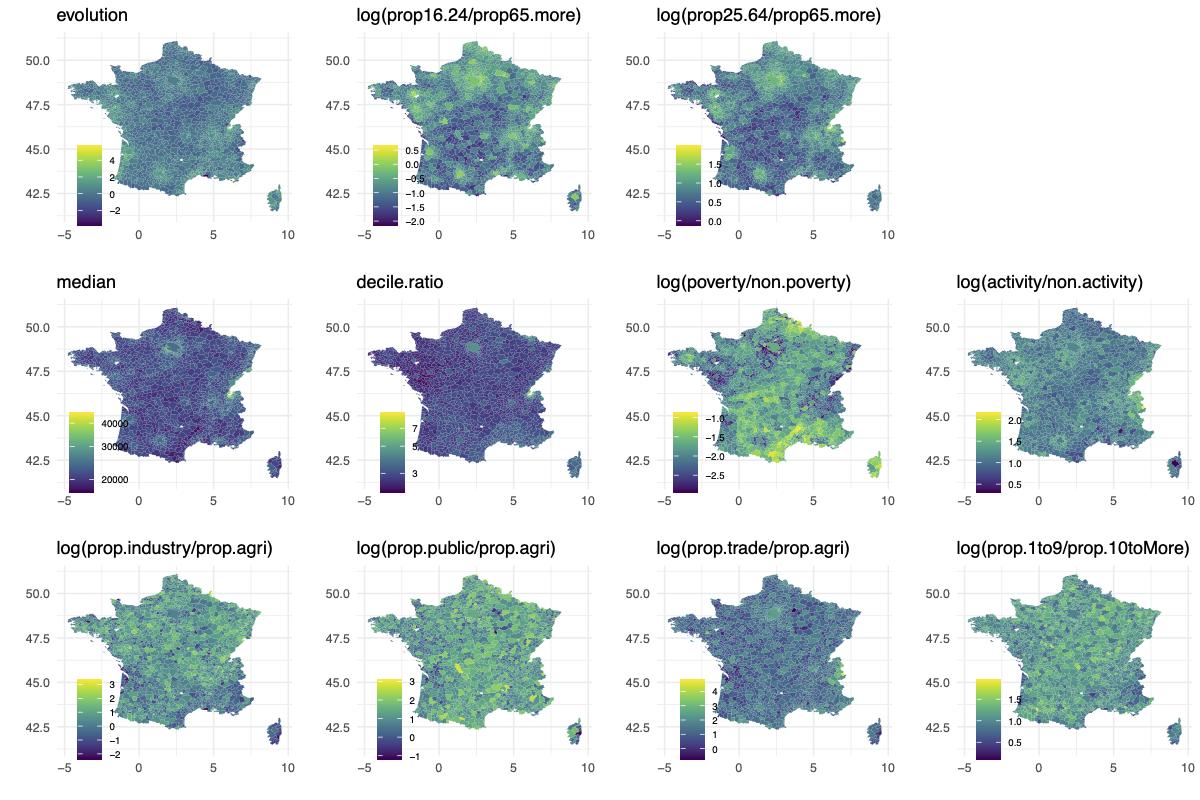}
    \caption{Spatial representation of covariates described in Section~\ref{sec:model} in mainland France. The spatial resolution is the living area (1681 living areas).}
    \label{fig:cov}
\end{figure}

\rev{The exploratory non parametric study of Section~\ref{sec:nonparametric} and in particular the analysis of the recentered $L$ and cross-$L$ functions have shown that the bivariate point pattern exhibits a significant departure from a bivariate inhomogeneous Poisson point process. We do not intend to model $(\bX_c,\bX_\ell)$, even its second-order structure. So, the likelihood of $(\bX_c,\bX_\ell)$ is not available. To sum up, our goal is to estimate $\bbeta_c$ and $\bbeta_\ell$ which characterize the intensities of $\bX_c$ and $\bX_\ell$ without neglecting that the points of $\bX_c$ (resp. $\bX_\ell$, resp. $(\bX_c,\bX_\ell)$) are dependent.}

These comments lead us to consider a composite likelihood method to estimate $\bbeta_c$ and $\bbeta_\ell$.
Univariate composite likelihoods for intensity estimation has a long history \citep{waagepetersen2009two, coeurjolly2019understanding,choiruddin2021information}. The simplest method is to use the Poisson likelihood which remains efficient even if the underlying process is not the Poisson model (this is the current situation here). Recent developments include regularized techniques \citep{choiruddin2018convex,choiruddin2023kppmenet} which can be applied in high-dimensional context (using a Lasso-type regularization) or with strongly correlated covariates (using a ridge penalty for instance). A contribution of the present paper is to investigate multivariate regularized estimation.

A separate exploratory analysis on the covariates has shown first that all of them seem to bring some information in the explanation of the presence/absence of a bank in a given living area and second that several covariates are highly correlated. Also, it could make sense to impose a condition on a whole set of parameters and in particular on covariates belonging to the same group, for example by imposing that the estimates for one group are shrunk to zero while keeping the possibility to select or not a feature using a Lasso type penalty. \rev{The possible multicollinearity and group structure of spatial covariates can be considered through a specific penalty term referred to as multivariate sparse group Lasso in the literature for generalized linear models (GLMs) \citep{li2015multivariate,liang2022sparsegl}. Our objective is to introduce a similar penalization to estimate the parameters of the bivariate model~\eqref{eq:marked}.}

\rev{In the current application we consider three groups of features}: covariates (i), (ii), (iii)  for lucrative and cooperative banks, so in total six groups. As a consequence, we propose to estimate {$\bbeta=(\bbeta_c^\top,\bbeta_\ell^\top)^\top = (\bbeta_1^\top,\bbeta_2^\top)^\top $} as the maximum of $Q(\bbeta)$, viewed as a penalized version of $\log$-Poisson likelihood of the bivariate point process $\mathbf X= (\mathbf X_c,\mathbf X_\ell)$. The composite likelihood denoted by $\ell(\bbeta)$ is given by 
\begin{align}
\ell(\bbeta)&= \sum_{i=1}^2 
\left\{
\sum_{x \in \mathbf X_i} \log \rho_i(x ; \bbeta_i) - \int_W \rho_i(x ; \bbeta_i) \mathrm d x 
\right\} \nonumber\\
&= 
\sum_{y\in \bY} \log \rho(y;\bbeta) - \int_{W\times \mathcal M} \rho(y;\bbeta) \dd y \label{eq:lik}
\end{align}
where $\bY=\{(\bX,i)_j,j=1,\dots\}$ represents the marked point process of all banks with type $i=1,2$, where $\mathcal M=\{1,2\}$ and where we abuse notation by denoting by $\dd y$ the Lebesgue measure times the counting one. Then, the normalized and regularized version $Q(\bbeta)$ are given by
\begin{equation}
    \label{eq:method}
Q(\bbeta) = \mu^{-1} \, \ell(\bbeta)
\quad -    \sum_{g=1}^G {\lambda^g} \|\bbeta^{g}\| -  \sum_{l=1}^p  \lambda_{l} |\beta_{l}|.
\end{equation}
In the definition of $Q$, we normalize $\ell(\bbeta)$ by $\mu = \E\{N(W\times \mathcal M)\} = \int_{W\times \mathcal M} \rho(y;\bbeta^\star)\dd y$ where $\bbeta^\star$ stands for the true parameter vector. The quantity $\mu$ plays the role of the sample size in this marked point process framework. Normalizing the composite likelihood $\ell$ by the sample size is also the convention considered by~\citet{friedman2010regularization} for GLMs. The notation $G$ stands for the number of groups. \rev{In our study, we consider covariate-based group structure, leading to $G=6$. The first two groups are demographic covariates for cooperative ($i=1$) and lucrative banks ($i=2$), the second two groups are social covariates for each type of bank, and the last two groups are economic covariates for each type of bank. In more details, $\bbeta^g$ is the vector of parameters corresponding to covariates of group $g$. For instance, $\bbeta^1=\{(\bbeta_{1})_j, j=2,\dots,4\}$ (the parameter vector for the group of demographic covariates for the cooperative banks) , $\bbeta^2=\{(\bbeta_{2})_j, j=2,\dots,4\}$, \dots, $\bbeta^6=\{(\bbeta_{2})_j, j=9,\dots,12\}$ (the parameter vector for the group of economic covariates for the lucrative banks).}
The parameters $\lambda^g$ and $\lambda_{l}$ are non-negative real numbers called regularization parameters. In~\eqref{eq:method}, the first penalty term tends to shrink all covariates belonging to the same group while the second one corresponds to the more standard adaptive Lasso penalty. We agree that from a practical point of view $\mu$ is unknown but it is worth pointing out that  maximizing $Q$ is equivalent to maximizing $\mu Q$ with $\lambda^g, \lambda_l$ replaced by new parameters $\check\lambda^g=\mu \lambda^g$ and $\check\lambda_l= \mu \lambda_l$.

Since the intensity function $\rho$ has an exponential form with baseline function $d(\cdot)$, then by exploiting the fact that all covariates are observed at the living area level (remind that the number of living areas is denoted by $J$), we have the following proposition.


\begin{proposition}\label{prop}
Maximizing $\ell(\bbeta)$ is equivalent to maximizing the $\log$-likelihood of two independent Poisson regressions with length $J$, with the canonical log link and with offset term $(\log \nu_j)_{j=1,\dots,J}$ where $\nu_j$ is the total number of inhabitants in the $j$th living area. Therefore, the maximization of $Q$ results in the maximization of a sparse group Lasso bivariate Poisson regression.
\end{proposition}
A consequence of Proposition~\ref{prop} is that we can follow~\citet{friedman2010regularization} and~\citet{liang2022sparsegl}, which implement various (adaptive) regularized versions of GLMs with exponential family models using coordinate descent algorithm.
Finally, our asymptotic result, namely Theorem~\ref{thm:convergence}, presented in Appendix~\ref{app:convergence} shows that the estimate $\hat \bbeta$ maximizing $Q(\bbeta)$ is, on the one hand, consistent when the average number of points $\mu$ increases and, on the other hand, satisfies an oracle property \citep[e.g.][]{fan2001variable}, that is $\mathrm P(\hat \beta_l=0) \to 1$ for any $l$ such that $\beta_l=0$. In other words, the procedure correctly identifies the covariates which are non informative.

\rev{For the tuning of regularization parameters $\lambda^g$ and $\lambda_l$, we extend a strategy developed by~\citet{zou2009adaptive,zhang2010regularization} and in particular used by~\citet{choiruddin2018convex, choiruddin2023adaptive} in the adaptive lasso case, combined with the default implementation of the sparse group lasso by~\cite{liang2022sparsegl}. We let $\lambda^g = (1-\alpha)\lambda / \|\hat{\bbeta}^g_{\mathrm{no}}\|$ and $\lambda_l=\alpha\lambda/|(\hat{\bbeta}_{\mathrm{no}})_l|$ where $\hat{\bbeta}_{\mathrm{no}}$ is a preliminary estimate of $\bbeta$, obtained without any regularization term. As recommended by~\citet{liang2022sparsegl}, we let $\alpha=0.05$. The last parameter $\lambda$ is obtained by minimizing the BIC criterion \citep[see][]{choiruddin2018convex}.}

\section{Results and discussion} \label{sec:results}

\rev{We randomly divide data into two sets (training set and validation set) with length $n/3$ and $2n/3$. The training dataset is used to tune the parameter $\lambda$. The validation set is used to perform the regularized composite likelihood with the tuned $\lambda$. This step allows us to identify seven covariates as non-informative. These are  represented by dots centered at 0 in Figure~\ref{fig:coeff}: \texttt{evolution} for cooperative and lucrative banks, \texttt{log(activity/non.activity)}, \texttt{log(prop.industry/prop.agri)} and\linebreak \texttt{log(prop1to9/prop.10toMore)}, for lucrative banks. The lack of evidence for the variable \texttt{evolution}  confirms the findings of previous studies, which suggest that the spatialization of banks is no longer demand-driven. Similarly, the distinction between agricultural and industrial workers is no longer relevant, given the decline in the number of agricultural workers in France. In the second step, we fit the bivariate composite likelihood to the entire dataset (training and validation) using only the covariates selected in the first step. We estimate the standard errors of the remaining 24-7=17 covariates using the covariance estimator developped by~\citet{coeurjolly2014covariance} and construct 90\% asymptotic confidence intervals (assuming a Gaussian distribution) for all parameters.}


\begin{figure}[htbp]
\centering
\includegraphics[width=.9\textwidth]{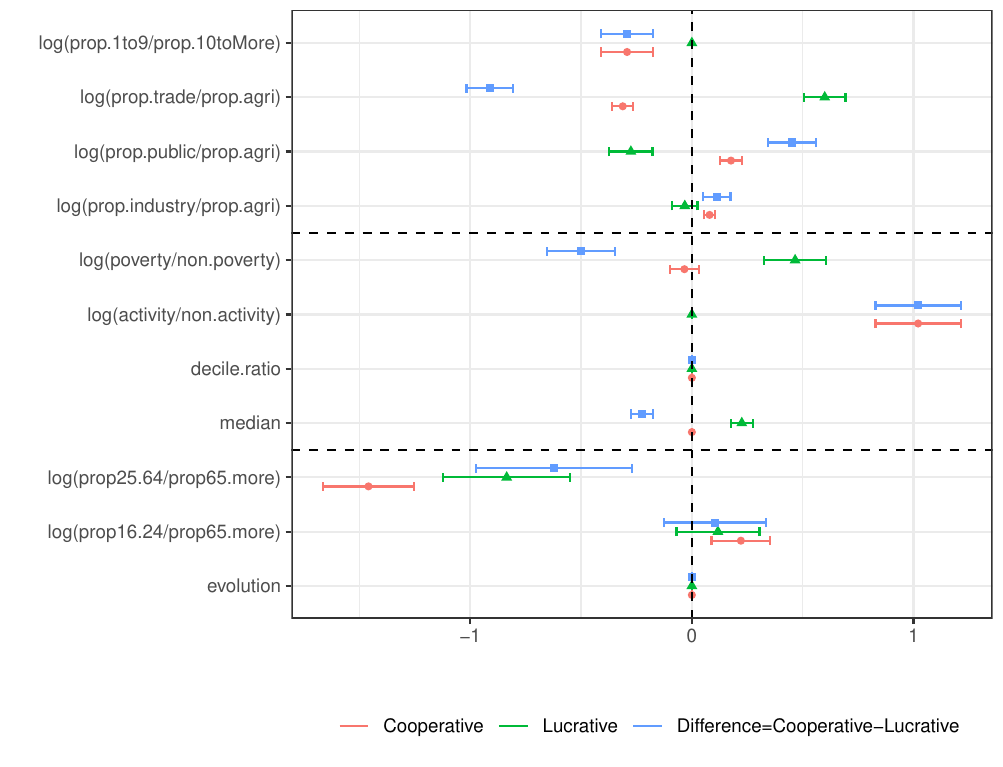}
\caption{Dots represent point estimates of parameters for each variable and each type of bank\rev{, as well as the difference between these two estimates}. Segments correspond to 90\% (asymptotic) confidence intervals. Dots centered at 0 mean that the related covariate is not selected by the sparse group lasso procedure. \rev{Estimates of intercept parameters are not represented.}}
\label{fig:coeff}
\end{figure}



\rev{Let us comment Figure~\ref{fig:coeff}.} Demographic covariates facilitate analysis of potential differences based on household demand. The results indicate minimal differentiation overall. Both types of banks exhibit similar trends across all variables. Population trends do not reveal significant differences. The proportion of the working population relative to the senior population \rev{(i.e. the covariate \texttt{log(prop25.64/prop65.more)})} negatively affects both types of banks, with a stronger \rev{significant} effect observed for cooperative banks. Additionally, the proportion of young people within the senior population \rev{(i.e. the covariate \texttt{log(prop16.24/prop65.more)})} is \rev{not significantly} more positively associated with cooperative banks.

\rev{Social covariates capture demand-side segmentation through three complementary dimensions of income — its level (median), its dispersion (\texttt{decile.ratio}), and labour-market participation (\texttt{log(activity/non.activity)}) — together with the local poverty rate (\texttt{log(poverty/non.poverty)}). This social dimension remains a source of differentiation between the two networks.  The activity rate for lucrative banks is insensitive to \texttt{log(activity/non.activity)}  whereas it has a positive effect for cooperative banks, consistent with their historical anchoring in working, member-based communities. Lucrative banks respond significantly to the level of income (\texttt{median}) and to the poverty rate (\texttt{log(poverty/non.poverty)}). 
The cooperative-lucrative 
differences confirm that income distribution and labour-market social covariates capture demand-side segmentation through three complementary dimensions of income — its level (\texttt{median}), its dispersion (\texttt{decile.ratio}), and labour-market participation (\texttt{log(activity/non.activity)}) — together with the local poverty rate (\texttt{log(poverty/non.poverty)}).}

Finally, economic covariates reveal more significant differences. The proportion of employees in industry versus agriculture \rev{(covariate \texttt{log(prop.industry/prop.agri)})} is not a \rev{significant} differentiating factor. While this is important for cooperative banks, the decline of agriculture in the French economy may explain why this variable has little influence. The proportion of SMEs \rev{(covariate \texttt{log(prop.1to9/prop.10toMore)})} distinguishes between the two networks: it has \rev{no} influence on profit-making banks, but a negative influence on cooperative banks. The two types of bank have opposite effects on the final two variables \rev{(covariates \texttt{log(prop.public/prop.agri)} and \texttt{log(prop.trade/prop.agri)})}. Profit-making banks are positively influenced by the presence of employees in the trade sector, while cooperative banks are negatively influenced. The latter are positively influenced by employees in the public sector.

These results demonstrate that, in line with our hypotheses, the spatialization of two types of banks can still be explained by customer segmentation on the demand side. Demographic distribution is no longer relevant, but the distribution of household income, the composition of the productive fabric, and job distribution remain informative and are sources of differentiation. 
However, these results do not corroborate our hypotheses regarding the supply side, in contrast to the demand side findings. The introduction of new parameters has instead solidified our understanding of how the spatialization of these two types of banks differs beyond what our supply-side hypotheses predicted.

We end this section with Figure~\ref{fig:fitint} which represents the fitted intensities for cooperative and lucrative banks \rev{as well as estimated relative risks, that is maps of\linebreak $\rho_\bullet(\cdot; \hat \bbeta_\bullet)/ \{\rho_c(\cdot;\hat\bbeta_c)+\rho_\ell(\cdot;\hat\bbeta_\ell)\}$ for $\bullet=c,\ell$. These figures confirm} our main hypothesis of different spatialization between the two types of bank. Integrating parametric dimensions enables us to link this differentiation to elements of demand in relation to income distribution, employment characteristics, and the local productive fabric.

\begin{figure}[htbp]
\centering
\includegraphics[width=.9\textwidth]{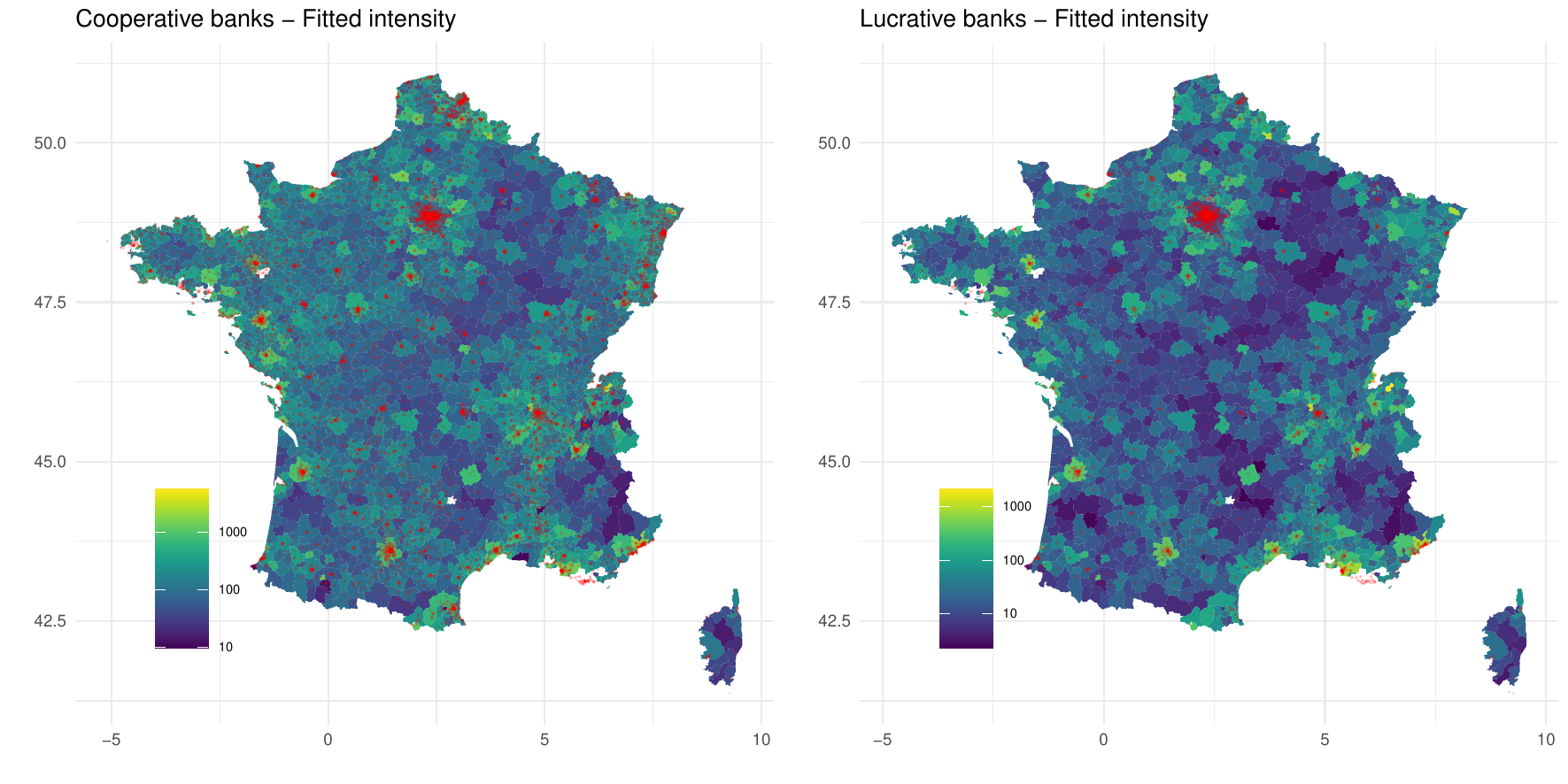}\\
\includegraphics[width=.9\textwidth]{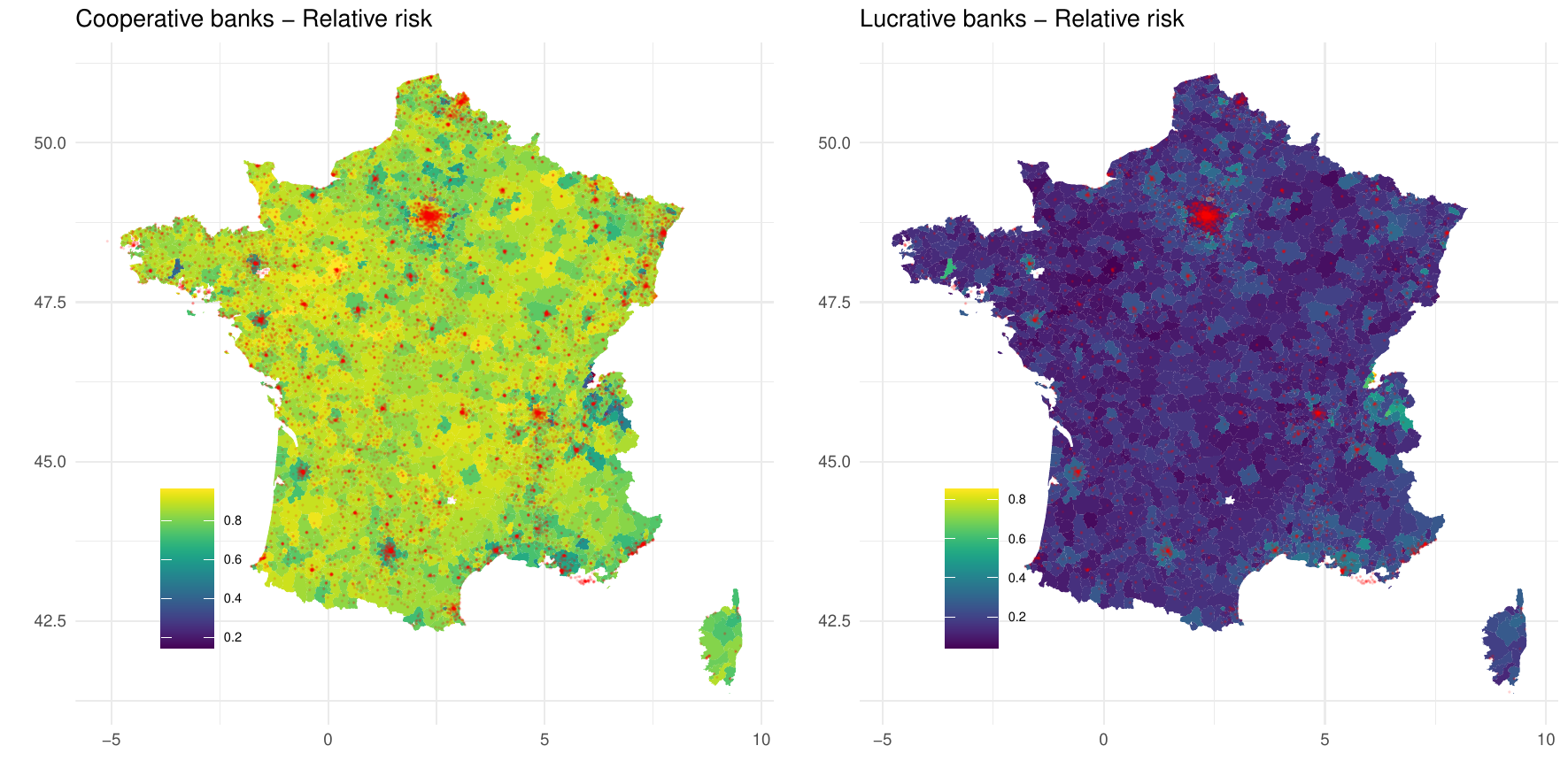}
\caption{\rev{Top row: fitted intensities, i.e. maps of $\rho_\bullet(\cdot;\hat \bbeta_\bullet)$; bottom row: estimated relative risks maps, i.e. maps of $\rho_\bullet(\cdot; \hat \bbeta_\bullet)/ \{\rho_c(\cdot;\hat\bbeta_c)+\rho_\ell(\cdot;\hat\bbeta_\ell)\}$ for $\bullet=c,\ell$. The intensity model corresponds to~\eqref{eq:method} and the parameters are estimated using the sparse group lasso procedure. Red dots correspond to locations of cooperative or lucrative banks.}} 
\label{fig:fitint}
\end{figure}

\section{Conclusion}

The following discussion  addresses the gaps that have been identified in the extant literature, specifically with regard to the intersection between the literature on the evolution of the banking sector and that on banking geography. The present study has been undertaken to examine the question of whether, by studying the spatialization of cooperative and profit-making banks, it is still possible to observe differences and test their origins. The expansion of digital commercial offerings by banking institutions does not appear to have resulted in a significant alteration of their spatialization. The search for proximity, materialised by the presence of branches, remains essential for resolving information asymmetries and producing idiosyncratic knowledge. 
The present contributions are twofold, in that they are both theoretical and methodological in nature. The demonstration is made of the persistence of differentiated spatialization between these two types of banks, which is explained by variables characterising the clienteles of these types of banks. For the first time, we propose a comparative spatial analysis based on unpublished data constructed during this research.
It was not possible to corroborate the hypothesis that organisational changes (supply logic) had an influence on the spatialization of banks. This remains an area of research that requires further consolidation. 

\rev{We identifty further research perspectives. First, it would be interesting to complete the present analysis (e.g. with a residual analysis) in order to seek for potential spatial missing confounders.}
\rev{Second, the exploratory non-parametric estimates of the intensity function revealed some areas of the country where many points are observed, in particular within the Région Ile-de-France. Focusing more specifically on urban areas, e.g. Paris, could be of great interest. This leaves the scope of the present study, since at the city level, it cannot be neglected that banks are living on a linear network (corresponding to streets). Point processes on linear networks are much more complicated to deal with from a computational point of view, see e.g. \citet{baddeley2017stationary}. We did not consider this specificity in this paper and have considered, as a first approximation, that banks could occur at any continuous location in mainland France.}
\rev{Third, our cross-sectional approach allows us to identify demographic, social, and economic factors linked to the observed regional differentiation between bank types. To further address temporal aspects, adding the age or establishment date of each branch as a covariate would help account for temporal layering. Incorporating such data could improve our understanding of the historical dynamics of branch creation and closure, making it a promising avenue for future research. }

\appendix

\section{Proof of Proposition~\ref{prop}} \label{app:prop}

\begin{proof}
Decompose the observation domain $W=\cup_{j=1}^J \mathcal A_j$ where $\mathcal A_j$ is the $j$th living area with volume $|\mathcal A_j|$. The baseline density function $d$ and all covariates are constant over $\mathcal A_j$ and we let $d_j=d(x)$, $\bz_1(x)=\bz_{1,j}$ and $\bz_2(x)=\bz_{2,j}$ for any $x\in \mathcal A_j$. Note that $\nu_j = d_j |\mathcal A_j|$. With these notation, we simply observe that the composite likelihood $\ell(\bbeta)$ can be rewritten as
\begin{align}
\ell(\bbeta) &= \sum_{i=1}^2 
\sum_{j=1}^J \left[
\sum_{x \in \bX_i \cap \mathcal A_j} \bbeta_i^\top \bz_{i}(x) - \int_{\mathcal A_j} d(x) \exp \left\{ \bbeta_i^\top \bz_{i}(x)\right\} \dd x
\right]\nonumber\\
&= \sum_{i=1}^2 
\sum_{j=1}^J \left\{
\sum_{x \in \bX_i \cap \mathcal A_j} \bbeta_i^\top \bz_{i,j} - \int_{\mathcal A_j} d_j \exp \left( \bbeta_i^\top \bz_{i,j}\right) 
\right\}
\nonumber\\
&= \sum_{i=1}^2 \sum_{j=1}^J 
\left\{
N_{i,j} \bbeta_i^\top \bz_{i,j} - \nu_j \exp(\bbeta_i^\top \bz_{i,j}) 
\right\} \nonumber\\
&= \sum_{i=1}^2 \sum_{j=1}^J
\left[
N_{i,j} \log \left\{\nu_j \exp( \bbeta_i^\top \bz_{i,j} ) \right\} - \nu_j \exp(\bbeta_i^\top \bz_{i,j}) 
\right] 
-\sum_{i=1}^2 \sum_{j=1}^J N_{i,j} \log \nu_j. \label{eq:glm}
\end{align}
Up to a term independent of $\bbeta$, the first double sum of~\eqref{eq:glm} precisely corresponds to the likelihood of two independent Poisson regressions with  the canonical link and with  offset term $(\log \nu_j)_{j=1,\dots,J}$.
\end{proof}

\section{Asymptotic result} \label{app:convergence}

We slightly extend the setting of Section~\ref{sec:parametric}. Assume we observe a multivariate point process $\bX=(\bX_1,\dots,\bX_M)$ with length $M$ (the application to the banks dataset corresponds to the case $M=2$) defined in $\mathbb R^d$ and observed in a bounded domain $W\subset \mathbb R^d$. Unambiguously, we can also consider $\bY$ the point process $\cup \bX_i$, marked with the random discrete mark in $\mathcal M=\{1,\dots,M\}$. We assume that the intensity of $\bY$ has the parametric form given by~\eqref{eq:int} where $d(\cdot)$ is some baseline function, $\bbeta_i, \bz_i(x) \in \mathbb R^{b_i}$ are respectively the parameter of interest and the vector of spatial covariates which can be either  piecewise constant or more general.

We let  $\bbeta=(\bbeta_1^\top,\dots,\bbeta_M^\top)^\top$ and $p= \sum_{i=1}^M b_i$ denote the full parameter vector and its length. For, $I\subseteq \{1,\dots,p\}$, we let $\bbeta_I = (\beta_l)_{l\in I}$.  Let $I_0=\{l : \beta_l=0\}$, $i_0=\#I_0$, $I_1=\{l : \beta_l\neq0\}$ and $i_1=\#I_1$. We also assume some extra information allowing us to group covariates $z_l$ for $l=1,\dots,p$. Consider the partition $\{1,\dots,p\} = \cup_{g=1}^G \mathcal G_g$ where $\mathcal G_g$ are disjoint sets with cardinality $\gamma_g$. We use the notation $\bbeta^g = \bbeta_{\mathcal G_g} = (\beta_l)_{l\in \mathcal G_g}$ and finally we let $\Gamma_0=\{g: \mathcal G_g\cap I_1=\emptyset\}$ and $\Gamma_1=\{1,\dots,G\}\setminus \Gamma_0$; $\Gamma_0$ represents the sets of groups for which all coefficients are inactive, ie equal to 0. We estimate $\bbeta^\star$, the true parameter vector with the sparse group lasso procedure given by~\eqref{eq:lik}-\eqref{eq:method}. The parameters $\lambda^g$ for $g=1,\dots,G$ and $\lambda_l$ for $l=1,\dots,p$ are non negative regularization parameters and we allow them to be stochastic. 

Our framework, stated by \ref{C:mu}, means that the average number of points increases with some index, say $n$. This setting includes the standard infill and increasing domain asymptotic frameworks, see~\citet{coeurjolly2025regularization} for a detailed discussion. In the rest of this section, $W$, $\mu$, $\bbeta$, $\bbeta^\star$, $\hat \bbeta$, $\ell(\bbeta)$, $Q(\bbeta)$, $p$, $i_0$, $i_1$, $\gamma_g$, $\ell$ depend on $n$. We avoid the notational dependence on $n$ to ease the reading. Our asymptotic results depend on the following conditions:
\begin{enumerate}[($\mathcal C$.1)]
\item $\mu\to \infty$ as $n\to \infty$. \label{C:mu}
\item $p^4/\mu \to 0$ as $n\to \infty$ and 
\begin{align*}
\sup_{x\in \R^d} |d(x)|<\infty, \qquad \sup_{n \ge 1} \sup_{l=1,\dots,p} \sup_{x\in \R^d} |z_l(x)|&<\infty,\\
 \int_{W\times \mathcal M} \int_{W\times \mathcal M} \rho(y;\bbeta)\rho(y^\prime;\bbeta) [g(y,y^\prime)-1]\dd y \dd y^\prime &= O(\mu) \\
\inf_{n\ge 1} \inf_{\bphi \in \R^p} \mu^{-1} \bphi^\top \{ - \ell^{(2)}(\bbeta^\star)\}\bphi&>0 \\
\sum_{i=1}^M\int_{W} \|\bz(x)\|^q \rho\{(x,i);\bbeta\} \dd x&= O_\P(p^{q/2} \mu)
\end{align*}
\label{C:MAS}
for $q=1,3$ and for any $\bbeta \in \R^p$ such that  $\|\bbeta - \bbeta^\star\|=O_\P(\sqrt{p/\mu})$. \rev{In the second condition, the function $g$ stands for the pair correlation function, given by $g(y,y^\prime)=\rho^{(2)}(y,y^\prime)/(\rho(y)\rho(y^\prime))$, with $\rho^{(2)}$ the second-order intensity function, see \citet{moller2003statistical}.}
\item As $n\to \infty$
$$
\left( \max_{l\in I_1} \lambda_l\right) \sqrt{\frac{i_1 \mu}{p}} =o_\P(1)
\qquad \text{ and } \qquad
\left( \max_{g\in \Gamma_1} \lambda^g\right) \sqrt{\frac{\left( \sum_{g\in \Gamma_1} \gamma_g\right) \mu}{p}} =o_\P(1).
$$ \label{C:sgl1}
\item As $n \to \infty$
$$
\left(\frac{1}{\inf_{l\in I_0} \lambda_l} \right) \; \sqrt{\frac{p^2}{\mu}} = o_\P(1).
$$ \label{C:sgl2}
\end{enumerate}

\rev{Condition~\ref{C:mu} expresses what we mean by increasing the amount of data. We assume that the mean number of points is a sequence tending to $\infty$ as $n\to \infty$. Condition~\ref{C:MAS} is very specific to the composite likelihood $\ell(\bbeta)$. This condition is an extension of the one considered by ~\citet{coeurjolly2025regularization} for univariate Poisson likelihood. It allows us to prove Lemma~\ref{lem} establishing the consistency of $\ell^{(1)}(\bbeta^\star)$. The first part of Condition~\ref{C:MAS}, $p^4/\mu\to0$, is probably the most informative: it means that the number of covariates can diverge to $\infty$ but not too quickly. Finally, Conditions~\ref{C:sgl1}-\ref{C:sgl2} are specific to the sparse group lasso penalty introduced in this paper and provide the necessary compromise to be considered, in order to ensure first, the consistency of $\hat \bbeta$ and second, an oracle type property. These conditions are extensions of the conditions on the regularization parameters imposed by~\citet{coeurjolly2025regularization} to prove that the univariate Poisson likelihood with adaptive Lasso penalization, for instance, leads to a consistent estimator and to a procedure which identifies informative covariates.}

\begin{lemma}
\label{lem} 
Under \ref{C:mu}-\ref{C:MAS},  the following statements hold. \\
(i) For any $\bbeta\in \R^p$ such that $\|\bbeta-\bbeta^\star\|=O_\P(\sqrt{p/\mu})$
\begin{equation}
\|\ell^{(1)}(\bbeta^\star)\| = O_\P(\sqrt{p\mu}) \quad 
\text{ and } \quad 
\frac{\partial \ell}{\partial \beta_l} (\bbeta)= O_\P(p\sqrt \mu). \label{eq:MAS1}
\end{equation}
(ii) For any $\bk\in \R^p$ (with $\|\bk\|\le K<\infty$) and $t\in (0,1)$,
\begin{equation}
\mu^{-1} \, \bk^\top \{-\ell^{(2)}(\bbeta^\star)\} \bk \ge \alpha  
\quad \text{ and } \quad 
\mu^{-1} \bk^\top \left\{ \ell^{(2)}(\bbeta^\star+t \bk \sqrt{p/\mu}) - \ell^{(2)}(\bbeta^\star)\right\} \bk =o_\P(1). \label{eq:MAS2}    
\end{equation}
\end{lemma}

\begin{proof}
(i) Using the marked point process formalism, the score function of the composite likelihood writes
$$
\ell^{(1)}(\bbeta) = 
\sum_{y=(x,i)\in \bY} \bz(x) - \sum_{i=1}^M\int_{W} \bz(x) \rho\{(x,i);\bbeta\}\dd x
$$
By Campbell theorem \citep{moller2003statistical}, $\E \left\{ \ell^{(1)}(\bbeta^\star\right\}=0$ and 
\begin{align*}
\Var \left\{ \ell^{(1)}(\bbeta^\star)\right\}=& 
\int_{W\times \mathcal M} \bz(x)\bz(x)^\top 
\rho(y,\bbeta^\star) \dd y \\
&+ 
\int_{W\times \mathcal M} \int_{W\times \mathcal M}
\bz(x) \bz(y)^\top
\rho(y;\bbeta^\star)\rho(y^\prime;\bbeta^\star) [g(y,y^\prime)-1]\dd y \dd y^\prime.
\end{align*}

By \ref{C:mu}-\ref{C:MAS}, we have that $\|\Var\left\{ \ell^{(1)}(\bbeta^\star)\right\}\|= O(p\mu)$ since in particular $\|\bz(x)\|^2=O(p)$ for any $x\in \R^d$. We deduce that $\|\ell^{(1)}(\bbeta^\star)\| = O_\P(\|\Var\left\{ \ell^{(1)}(\bbeta^\star)\right\}\|^{1/2}) = O_\P(\sqrt{p\mu}).$ 
For the second part of~\eqref{eq:MAS1}, using again Campbell Theorem, \ref{C:MAS} and Taylor expansion we have for some $\tilde \bbeta=\bbeta+t(\bbeta^\star-\bbeta)$ for $t\in(0,1)$
\begin{align*}
\left|\E \left\{\frac{\partial \ell}{\partial \beta_l} (\bbeta)\right\} \right| &\le \int_{W\times \mathcal M} 
\left| 
z_l(x) \left\{ \rho(y;\bbeta)-\rho(y;\bbeta^\star)\right\} 
\right|\dd y\\
&\le \kappa \int_{W\times \mathcal M} \|\bz(x)\| \; \|\bbeta-\bbeta^\star\| \rho(y ; \tilde\bbeta) \dd y \\
&= O_\P ( p \sqrt{\mu}).
\end{align*}
Using the same tools, we can check that 
$\Var \left\{
\frac{\partial \ell}{\partial \beta_l} (\bbeta)
\right\} = O(\mu)$ whereby we deduce the result.\\
(ii) The first part directly ensues from the last assumption of~\ref{C:MAS}. The sensitivity matrix at $\bbeta$ writes
$$  
-\ell^{(2)}(\bbeta) = \sum_{i=1}^M \int_W \bz(x)\bz(x)^\top \rho\{(x,i);\bbeta\} \dd x
$$
Let $\Delta(\bk)$ be the second part term of~\eqref{eq:MAS2}. Using a Taylor expansion and~\ref{C:MAS} we have, by denoting $\tilde \bbeta = \bbeta + s t \bk \sqrt{p/\mu}$ for $s\in(0,1)$
\begin{align*}
\|\Delta(\bk)\| &\le \mu^{-1}\sum_{i=1}^M \int_W \|\bz(x)\|^2 
\; \|st \bk \sqrt{p/\mu}\| 
\; \|\bz(x)\| \; \rho\{(x,i), \tilde \bbeta\} \dd x  \\
&\le K \mu^{-1}\sqrt{p/\mu} \sum_{i=1}^M\int_W \|\bz(x)\|^3 \rho\{(x,i), \tilde \bbeta\} \dd x \\
& = O_\P\left( \sqrt{\frac{p^4}{\mu}}\right) =o_\P(1).
\end{align*}
\end{proof}


\begin{theorem} \label{thm:convergence} Let $\hat \bbeta = \mathrm{argmax}_{\bbeta}  \; Q(\bbeta)$. Assume that  the conditions \ref{C:mu}-\ref{C:sgl2} hold, then (i) 
$\hat\bbeta -\bbeta^\star= O_{\P} ( \sqrt{p/\mu})$ and (ii) $\mathrm{P}(\hat \bbeta_{I_0}=0) \to 1$ as $n \to \infty$.
\end{theorem}

\begin{proof}
(i) Let $\bk \in \R^p$ and $\Delta(\bk)= Q(\bbeta^\star+\sqrt{p/\mu})-Q(\bbeta^\star)$. The consistency result is proved if one proves that for given $\varepsilon>0$, there exists $K$ such that for $n$ large enough $\P\{\sup_{\|\bk\|=K} \Delta(\bk)\}\le \varepsilon$. Decompose $\Delta(\bk)=T_1+T_2+T_3$ with
\begin{align*}
T_1 &= \mu^{-1} \left\{ \ell(\bbeta^\star+\sqrt{p/\mu} \bk) -\ell(\bbeta^\star)\right\}\\
T_2 &= \sum_{g=1}^G \lambda^g \left( \|\bbeta^{g,\star}\| - \|\bbeta^{g,\star}+ \sqrt{p/\mu} \, \bk^g\|\right) \\
T_3&= \sum_{l=1}^p \lambda_l \left(
|\beta_l^\star|-|\beta_l^\star+  \sqrt{p/\mu}\, k_l|  
\right)
\end{align*}
where we use the notation $\bk=(k_l)_{l\in\{1,\dots,p\}}$ and $\bk^g=(k_l)_{l\in \mathcal G_g}$. From Lemma~\ref{lem}, there exists $\alpha>0$ (independent of $\bk$ and $n$) such that
$$
T_1 = \frac{1}\mu \sqrt{\frac{p}{\mu}} \bk^\top \ell^{(1)}(\bbeta^\star) - \frac{\alpha}2 \, \frac{p}{\mu} \, \|\bk\|^2 + R_1 
$$
for $n$ large enough, with $R_1=o_\P(p/\mu)$. Regarding the term $T_3$, we have
\begin{align*}
T_3 & \le \sum_{l\in I_1} \lambda_l \left(
 |\beta_l^\star|-
 |\beta_l^\star+  \sqrt{p/\mu}\, k_l| 
 \right)
 \\
& \le \left( \max_{l \in I_1} \lambda_l\right) \sqrt{\frac{p}\mu} \,  \sum_{l \in I_1} |k_l| \\
&\le \left( \max_{l \in I_1} \lambda_l\right) \sqrt{\frac{p}\mu} \, \sqrt{i_1} \, \|\bk\|.
\end{align*}
Finally, for the term $T_2$ we proceed similarly to derive (for $n$ large enough)
\begin{align*}
T_2 &\le \sum_{g \in \Gamma_1} \lambda^g \left(
\|\bbeta^g\| - \|\bbeta^{g,\star}+ \sqrt{p/\mu} \, \bk^g\|\right)\\
&\le 4 \left(\max_{g \in \Gamma_1} \lambda^g \right) \sqrt{\frac{p}\mu} \, \sum_{g\in \Gamma_1} \|\bk^g\| \\
& \le 4 \left(\max_{g \in \Gamma_1} \lambda^g \right) \sqrt{\frac{p}\mu} 
\left( \sum_{g\in \Gamma_1} \gamma_g\right)^{1/2} \, \|\bk\|.
\end{align*}
Thus for $n$ large enough
$$
\rev{\Delta(\bk)} \le \frac1\mu \, \sqrt{\frac{p}{\mu}} \, \bk^\top \ell^{(1)}(\bbeta^\star)\bk \, - \, \frac{\alpha}2 \, \frac{p}{\mu} \|\bk\|^2 \, + \, \omega
$$
with 
$$
\omega = R_1+4 K\left(\max_{g \in \Gamma_1} \lambda^g \right) 
\sqrt{\frac{\left(\sum_{g\in \Gamma_1} \gamma_g\right)p}{\mu}} 
\; +\;
K\left( \max_{l \in I_1} \lambda_l\right) \sqrt{\frac{i_1 p}\mu}
$$
whereby we deduce that for $n$ large enough
$$
\P \left\{ \sup_{\|\bk\|=K} \Delta(\bk)\right\} \le \P 
\left( L \ge \frac{\alpha}2 K\sqrt{p\mu}\right)
$$
where $L=\|\ell^{(1)}(\bbeta^\star)\| + \mu \sqrt{\frac{\mu}{p}} \omega$. By condition~\ref{C:sgl1}, $\omega=o_\P(p/\mu)$. This and~\eqref{eq:MAS1} show that $L=O_\P(\sqrt{p\mu})$ which ends the proof since we can choose $K$ large enough.\\

\noindent (ii) Following (i), the pioneering work by~\citet{fan2001variable} and the more recent work by~\citet{coeurjolly2025regularization} (for point process intensity estimation), the result is proved if one proves that for any $\bbeta\in \R^p$ such that $\|\bbeta_{I_1}-\bbeta_{I_1}^\star\|=O_\P(\sqrt{p/\mu})$, we have for any $l\in I_0, \kappa>0$ and $\varepsilon=\kappa \sqrt{p/\mu}$
$$
\frac{\partial Q}{\partial \beta_l}(\bbeta)<0 \; \text{ for } 0 <\beta_l<\varepsilon 
\quad \text{ and } \quad
\frac{\partial Q}{\partial \beta_l}(\bbeta)>0 \; \text{ for } -\varepsilon <\beta_l<0.
$$
We consider the first statement as the other one follows along similar lines. Let $\gamma:\{1,\dots,p\} \to \{1,\dots,G\}$ be the function given by 
$\gamma(l)=\sum_{g=1}^G g \mathbf 1(l\in \mathcal G_g)$. We have
\begin{align*}
\frac{\partial Q}{\partial \beta_l}(\bbeta) &= \frac1\mu \frac{\partial \ell}{\partial \beta_l}(\bbeta) - \lambda_l \; \mathrm{sign}(\beta_l) - \sum_{g=1}^G \lambda^g \frac{\beta_l}{\|\bbeta^g\|} \mathbf 1(l\in \mathcal G_g) \\
&= \frac1\mu \frac{\partial \ell}{\partial \beta_l}(\bbeta) - \lambda_l \; \mathrm{sign}(\beta_l) - \lambda^{\gamma(l)} \frac{\beta_l}{\|\bbeta^{\gamma(l)}\|} \\
&\ge \frac1\mu \frac{\partial \ell}{\partial \beta_l}(\bbeta) - m(\bbeta)
\end{align*}
where $m(\bbeta)=\inf_{l \in I_0} \left( \lambda_l + \lambda^{\gamma(l)}\frac{\beta_l}{\|\bbeta^{\gamma(l)}\|}\right)$. The  lower-bound of $m(\bbeta)$, uniform in $\bbeta$, is $\inf_{l\in I_0} \lambda_l$. Thus, we have
\begin{align*}
\P \left\{ \frac{\partial Q}{\partial \beta_l}(\bbeta) <0\right\}
&\ge \P \left\{
 \frac{\partial \ell}{\partial \beta_l}(\bbeta) \le \mu \; \inf_{l\in I_0}\lambda_l \right\}.
\end{align*}
From Lemma~\ref{lem} and in particular~\eqref{eq:MAS1}, the result follows if $(\inf_{l\in I_0} \lambda_l) \sqrt{\mu/p^2}\to \infty$, which ensues from condition~\ref{C:sgl2}.
\end{proof}

\bibliographystyle{abbrvnat} 
\bibliography{refs}

\end{document}